\documentclass[a4paper,USenglish,cleveref,thm-restate]{lipics-v2021}
\pdfoutput=1

\usepackage{style}
\usepackage{shortcuts}

\title{A Simple Algorithm for Trimmed Multipoint Evaluation}

\author{Nick Fischer}{INSAIT, Sofia University ``St.\ Kliment Ohridski'', Bulgaria}{nick.fischer@insait.ai}{https://orcid.org/0009-0001-0909-3296}{Partially funded by the Ministry of Education and Science of Bulgaria's support for INSAIT as part of the Bulgarian National Roadmap for Research Infrastructure. Part of this work was done while the author was affiliated with the Weizmann Institute of Science.}
\author{Melvin Kallmayer}{Goethe University Frankfurt, Germany}{mel.kallmayer@gmail.com}{https://orcid.org/0009-0001-2466-8246}{}
\author{Leo Wennmann}{University of Southern Denmark, Odense, Denmark}{wennmann@imada.sdu.dk}{https://orcid.org/0009-0001-3346-6494}{Supported by Dutch Research Council (NWO) project "The Twilight Zone of Efficiency: Optimality of Quasi-Polynomial Time Algorithms" [grant number OCEN.W.21.268].}
\authorrunning{N.~Fischer, M.~Kallmayer and L.~Wennmann} 
\Copyright{Nick Fischer, Melvin Kallmayer, Leo Wennmann}



\EventEditors{Anne Benoit, Haim Kaplan, Sebastian Wild, and Grzegorz Herman}
\EventNoEds{4}
\EventLongTitle{33rd Annual European Symposium on Algorithms (ESA 2025)}
\EventShortTitle{ESA 2025}
\EventAcronym{ESA}
\EventYear{2025}
\EventDate{September 15--17, 2025}
\EventLocation{Warsaw, Poland}
\EventLogo{}
\SeriesVolume{351}
\ArticleNo{87}

\allowdisplaybreaks
\nolinenumbers

\begin{document}

\maketitle

\begin{abstract}
Evaluating a polynomial on a set of points is a fundamental task in computer algebra. In this work, we revisit a particular variant called \emph{trimmed} multipoint evaluation: given an $n$-variate polynomial with bounded individual degree $d$ and total degree $D$, the goal is to evaluate it on a natural class of input points. This problem arises as a key subroutine in recent algorithmic results [Dinur; SODA '21], [Dell, Haak, Kallmayer, Wennmann; SODA '25]. It is known that trimmed multipoint evaluation can be solved in near-linear time [van der Hoeven, Schost; AAECC '13] by a clever yet somewhat involved algorithm. We give a \emph{simple} recursive algorithm that avoids heavy computer-algebraic machinery, and can be readily understood by researchers without specialized background.
\end{abstract}
\section{Introduction} \label{sec:introduction}
One of the most fundamental problems in computer algebra is to efficiently evaluate a polynomial $P$ on some set of points, known as the \emph{multipoint evaluation} problem. Besides its importance as one of the most basic algebraic primitives, this problem finds many further applications in computer algebra (such as modular composition and polynomial factorization) and in algorithm design beyond (in diverse fields such as computational geometry, coding theory and cryptography). The inverse task, to \emph{interpolate} a polynomial from a given set of evaluations, is an equally important primitive.

For univariate polynomials, a textbook algorithm~\cite{BorodinM74,vonzurGathenG13} solves the multipoint evaluation problem in near-linear time. This algorithm generalizes to multivariate polynomials~\cite{Pan94} (via a simple divide-and-conquer method sometimes referred to as \emph{Yates' algorithm}~\cite{Yates37}), however, only in the restricted setting where the evaluation points form a cartesian \emph{grid}. Specifically, Yates' algorithm evaluates an $n$-variate polynomials with individual degree $d$ on all points of a grid
\begin{equation*}
    Z = \set{z_{1, 0}, \dots, z_{1, d}} \times \dots \times \set{z_{n, 0}, \dots, z_{n, d}}
\end{equation*}
in near-optimal time\footnote{Here and throughout, we write $O^*(\cdot)$ to omit polynomial factors in $n$ and $d$.} $O^*((d+1)^n)$. Lifting this restriction to grid points has been the focus of a long and active line of research~\cite{NuskenZ04,Umans08,KedlayaU11,HoevenS13,BjorklundKW19,HoevenL21a,BhargavaG0M22,BhargavaGG0U22} which, following breakthroughs by Umans~\cite{Umans08} and Kedlaya and Umans~\cite{KedlayaU11}, only recently culminated in an algorithm with almost-optimal running time $(d + 1)^{(1+o(1))n} \poly(n, d, \log |\mathbb F|)$, for all finite fields $\mathbb F$ and for $(d + 1)^n$ arbitrary evaluation points, due to Bhargava, Ghosh, Guo, Kumar and Umans~\cite{BhargavaGG0U22}. This fully settles the multipoint evaluation problem for \emph{dense} polynomials (over finite fields), but leaves open whether almost-linear time can also be achieved for (some classes of) \emph{sparse} polynomials and evaluation points.

\medskip
In this paper we focus on one natural such class, called \emph{trimmed} multipoint evaluation, with important applications in the design of exact and parameterized algorithms. Trimmed multipoint evaluation can be solved in near-linear time by an algorithm due to van der Hoeven and Schost~\cite{HoevenS13}. Our contribution is that we make this result accessible to modern algorithm design (beyond computer algebra) by distilling a particularly \emph{simple} recursive algorithm.

\paragraph*{Trimmed Multipoint Evaluation}
In the trimmed multipoint evaluation problem we focus on the class of $n$-variate polynomials~$P$ with individual degree $d$ \emph{and total degree\footnote{Recall that the \emph{individual} degree $d$ of a polynomial is the largest exponent of a variable in a monomial, whereas the \emph{total} degree $D$ is the largest sum of exponents in a monomial. E.g., $X_1^2 X_2$ has individual degree $d = 2$ and total degree $D = 3$.} $D$}. This is a very natural class of polynomials which, for $D \ll n d$, is exponentially sparser than polynomials with just an individual degree bound. As evaluation points we consider triangular subsets of grids $Z$ defined by
\begin{equation*}
    \set{(z_{1, \ell_1}, \dots, z_{n, \ell_n}) : \ell \in \set{0, \dots, d}^n,\, \ell_1 + \dots + \ell_n \leq D} \subseteq Z,
\end{equation*}
to which we will informally refer as \emph{trimmed} grids. This is arguably the most naturally matching class of evaluation points. To see this, first observe that the number of relevant (i.e., possibly nonzero) coefficients of $P$ equals exactly the number of grid points. We denote this number by $\ebc{n}{\leq D}{d}$ (which can be seen as an appropriate generalization of a binomial coefficient called an \emph{extended} binomial coefficient). More importantly, it turns out that the polynomial $P$ is \emph{uniquely} determined by the evaluations on an appropriate trimmed grid, i.e., we can interpolate $P$ given only these evaluations. Trimmed grids are, in a sense, the only sets of evaluation points satisfying this property; see~\cite{Sauer04}.

\paragraph*{Applications}
Besides being a natural problem in its own right, our interest in trimmed multipoint evaluation stems mainly from its applications in the context of exponential-time algorithms, most notably in a line of research on solving systems of polynomial equations~\cite{LokshtanovPTWY17,BjorklundK019,Dinur21,DellHKW25}. In this problem the input consists of $n$-variate polynomials~$P_1, \dots, P_m$ over some finite field $\mathbb F_q$ with (total) degree~$\Delta$, and the task is to test if all polynomials simultaneously vanish at some point~$x \in \mathbb F_q^n$, i.e.,~\makebox{$P_1(x) = \dots = P_m(x) = 0$}. The initial breakthrough due to Lokshtanov, Paturi, Tamaki, Williams and Yu~\cite{LokshtanovPTWY17} established that this problem can be solved exponentially faster than brute-force, in time $(q - \epsilon)^n$ for some $\epsilon > 0$, whenever $q$ and $\Delta$ are constant. This inspired several follow-up papers aiming to optimize the precise exponential running time~\cite{BjorklundK019,Dinur21,DellHKW25}. Trimmed multipoint evaluation shows up as a critical subroutine in the algorithms due to Dinur~\cite{Dinur21} (for multilinear polynomials, $d = q-1 = 1$) and due to Dell, Haak, Kallmayer and Wennmann~\cite{DellHKW25} (for general~\makebox{$d = q-1 \geq 1$}). Solving systems of polynomial equations in turn has many more applications on both the theoretical side---e.g., parity-counting directed Hamiltonian cycles~\cite{BjorklundH13}---and the practical side---e.g., the security of several so-called multivariate cryptosystems is based on the hardness of solving quadratic equations~\cite{Dinur21,Patarin96,KipnisPG99}. Especially for the latter it could be interesting to achieve simple, implementable algorithms.

In another closely related work, Björklund, Husfeldt, Kaski and Koivisto~\cite{BjorklundHKK10} considered trimmed\footnote{\setstretch{1.1}In fact, they consider a more general definition of ``trimmed'' allowing arbitrary downward-closed sets.} variants of the Zeta and Möbius transforms to develop fast algorithms for various exponential-time graph problems, such as computing the chromatic number for constant-degree-bounded graphs. The Zeta and Möbius transforms can be regarded as special cases of polynomial multipoint evaluation.\footnote{\setstretch{1.1}Indeed, recall that the zeta transform of a function $f : \set{0, 1}^n \to \mathbb F$ is defined as \smash{$(f \zeta)(X) = \sum_{Y \subseteq X} f(Y)$}, where we identify sets $X \subseteq [n]$ with their indicator vectors $X \in \set{0, 1}^n$. Consider the $n$-variate polynomial \smash{$P(x_1, \dots, x_n) = \sum_{Y \subseteq [n]} f(Y) \prod_{i \in Y} x_i$}, and observe that $(f \zeta)(X) = P(X)$. Thus, the zeta transform $f \zeta$ can be read off the evaluations of $P$ on the grid \smash{$\set{0, 1}^n$}.}

\paragraph*{State of the Art}
The state of the art for trimmed multipoint evaluation, as mentioned before, is a clever algorithm with near-linear running time $O^*(\ebc{n}{\leq D}{d})$ due to van der Hoeven and Schost~\cite{HoevenS13}. Their algorithm is based on the classical concept of Newton interpolation (see e.g.~\cite{BiniP94}), suitably tailored to the trimmed problem (see also~\cite{Werner80,Boor92,Sauer04,ChkifaCS14,DynF14} for some more references with a more mathematical point of view). The algorithm also offers two additional benefits: (1) It even solves a strictly more general multipoint evaluation problem on arbitrary downward-closed sets of relevant coefficients and grid points. (2) van der Hoeven and Schost have optimized the lower-order factors achieving an algebraic algorithm with $O(n N \log^2 N \log\log N)$ field operations where $N = \ebc{n}{\leq D}{d}$. The potentially remaining lower-order improvements even persist in the univariate setting. Thus, all in all, the trimmed multipoint evaluation problem has already been satisfyingly resolved.

The only downside is that van der Hoeven and Schost's algorithm is arguably somewhat intricate---both in the sense that it can be technically demanding to understand, particularly for researchers outside the computer algebra community, and in that it relies on several textbook algebraic primitives, such as efficient conversions between polynomial bases~\cite{BiniP94} and the use of truncated Fourier transforms as an implementation detail to achieve further improvements.

\paragraph*{Our Contribution}
Our focus here is \emph{not} to optimize the running time or generality of this state of the art. Instead, given the many exciting algorithmic applications our contribution is to make van der Hoeven and Schost's result for trimmed multipoint evaluation accessible to the algorithms community. We design a \emph{simple} recursive algorithm that is teachable to researchers without any background in computer algebra. Moreover, our algorithm does not rely on any black-box algebraic primitives other than Gaussian elimination. As in~\cite{HoevenS13}, we also obtain an equally simple algorithm for the interpolation problem. 

Our emphasis on simplicity clashes, however, with the additional benefits~(1) and~(2): (1)~It seems hard to obtain a recursive algorithm for the more general problem, and (2) optimizing lower-order factors would involve dealing with more details. Besides, for exponential-time algorithms we are typically anyway not bothered with keeping track of polynomial factors. For these reasons we have decided to stick to the simplest version, resulting in a pleasingly simple 8-line algorithm.

\paragraph*{Remark on Terminology}
The term ``trimmed'' multipoint evaluation is inspired by the related algorithms for Zeta and Möbius transforms in~\cite{BjorklundHKK10}. It is not standard in the mathematical literature, where one more often encounters terms like ``triangular subsets of tensor product grids''. We adopt the ``trimmed'' terminology here as this paper is primarily intended for the algorithms community.
\section{Preliminaries} \label{sec:preliminaries}
We write $[n] = \set{1, \dots, n}$.
Throughout, let $\mathbb F$ be a field and assume that we can evaluate field operations in unit time. 
For integers $n, k, d$ with $n \geq 0$ and $d \geq 1$ we define the \emph{extended binomial coefficient}
\begin{equation*}
    \ebc{n}{k}{d} = \abs{\set{ \ell \in \set{0, \dots, d}^n : \ell_1 + \dots + \ell_n = k}}.
\end{equation*}
That is, $\ebc{n}{k}{d}$ counts the number of multisubsets of $[n]$ with size $k$ and multiplicity at most~$d$, or equivalently, the number of monomials with total degree $k$ and individual degree at most~$d$ in an $n$-variate polynomial. In the same spirit we define the set
\begin{equation*}
    \ebc{[n]}{k}{d} = \set{ \ell \in \set{0, \dots, d}^n : \ell_1 + \dots + \ell_n = k}.
\end{equation*}
As a shorthand, we write $\ebc{n}{\le k}{d} = \sum_{i=0}^k \ebc{n}{i}{d}$,
and similarly define $\ebc{[n]}{\leq k}{d} = \bigcup_{i=0}^k \ebc{n}{i}{d}$. 
Further, we rely on the following generalization of Pascal's triangle; see~\cite{Bollinger93}.

\begin{lemma}[Extended Pascal Triangle] \label{lem:extended-pascal}

For $n, d \geq 1$ it holds that
\begin{equation*}
    \ebc{n}{k}{d} = \sum_{j=0}^d \ebc{n-1}{k-j}{d}.
\end{equation*}
\end{lemma}
    
Throughout we refer to a set $Z = \set{z_{1, 0}, \dots, z_{1, d}} \times \dots \times \set{z_{n, 0}, \dots, z_{n, d}}$ of field elements~$z_{i, j}$ as a \emph{grid}. To conveniently refer to the grid points we regularly write $Z_\ell = (z_{1, \ell_1}, \dots, z_{n, \ell_n})$ for $\ell \in \set{0, \dots, d}^n$. In particular, the subset of grid points $Z_\ell$ where $\ell$ ranges over $\ebc{[n]}{\leq D}{d}$ constitutes exactly a trimmed grid as introduced before.
\section{Trimmed Multipoint Evaluation and Interpolation} 
\label{sec:algo}

In this section, we present a simple, algebraic algorithm for trimmed multipoint evaluation.
As it is usually the case, the same algorithmic approach yields a simple algorithm for trimmed interpolation as well.

\subsection{Key Ideas}
\label{subsec:key-ideas}

The problem of evaluating univariate polynomials can be defined via the Vandermonde matrix.
For $z_{0}, \dots , z_{d} \in \F$, define the (square) Vandermonde matrix $V \in \F^{(d+1) \times (d+1)}$ as
\begin{equation*}
    V =
    V(z_{0}, \dotsc , z_{d}) = 
    \begin{pmatrix}
        1 & z_{0} & z_{0}^2 & \cdots & z_{0}^d\\
        1 & z_{1} & z_{1}^2 & \cdots & z_{1}^d\\
        \vdots & \vdots & \vdots & \ddots & \vdots\\
        1 & z_{d} & z_{d}^2 & \cdots & z_{d}^d
    \end{pmatrix}.
\end{equation*}
Let $a \in \F^{d+1}$ be the coefficient vector of the univariate polynomial $P(X) = \sum_{i=0}^{d} a_iX^i$ of degree~$d$.
The evaluation of~$P$ at all points~$z_{j}$ with $j \in \set{0, \dotsc, d}$ can be expressed as the matrix-vector product $V \cdot a = y$ (indeed, each entry is of the form $y_j = \sum_{i=0}^{d} a_i z_{j}^i = P(z_{j})$).
Similarly, we can interpolate~$P$ from its evaluations~$P(z_{j})$ using the matrix-vector product~$V^{-1} \cdot y = V^{-1} \cdot V \cdot a = a$.
The generalization to multivariate polynomials is simple:
Let~$a \in \F^{(d+1)^n}$ be the coefficient vector of an $n$-variate polynomial~$P$ with individual degree~$d$ that we want to evaluate on all grid points in~$Z$.
Then the Kronecker product~$V(z_{1,0}, \dotsc , z_{1,d}) \otimes \cdots \otimes V(z_{n,0}, \dotsc , z_{n,d}) \cdot a = y$ yields the vector~$y \in \F^{(d+1)^n}$ of all grid point evaluations.
The classical algorithm of Yates \cite{Yates37} allows to compute these evaluations recursively.
Notably, this computation does \emph{not} account for the total degree~$D$ and automatically results in a running time of $(d+1)^n$ for both multipoint evaluation and interpolation---even in our setting where the total degree~$D$ is much smaller than its maximal value~$nd$.
Consequently, we need more insights to tailor this approach to the trimmed requirements.

\subparagraph*{Idea 1: Recursion Scheme.}
Our goal is to achieve a running time that is linear in~$\ebc{n}{\le D}{d}$, i.e., in the number of grid points on which we want to evaluate a multivariate polynomial~$P$.
As a baseline, we start with a short explanation of Yates' algorithm: 
write $P(X_1, \dots , X_{n}) = \sum_{i= 0}^{d} P_i(X_1, \dots , X_{n-1}) \cdot X_{n}^i$ to obtain $d+1$ many $(n - 1)$-variate polynomials~$P_i$ of degree~$D-i$, where each~$P_i$ can be seen as a coefficient of a univariate polynomial in~$X_n$.
By making $d+1$ many recursive calls of size~$\ebc{n-1}{\le D}{d}$, we compute the evaluations~$P_i(Z_{\ell'})$ for all~$\ell' \in \ebc{[n-1]}{\le D}{d}$.
We obtain all~$P(Z_\ell)$ for~$\ell \in \ebc{[n]}{\le D}{d}$ by evaluating the univariate polynomials~$\sum_{i=0}^{d} P_i(Z_{\ell'}) \cdot X_n^i$ at all grid points~$z_{n, i}$.
However, this does \emph{not} exploit the degrees~$D-i$ of the polynomials~$P_i$ in the recursive calls, and hence results in the running time of~$(d+1)^n$.

Consider the following identity of the extended binomial coefficient that can be derived from \cref{lem:extended-pascal}
\begin{equation*}
    \label{eq:extended-pascal-2}
    \ebc{n}{\le D}{d} = \ebc{n-1}{\le D}{d} + \ebc{n-1}{\le D-1}{d} + \cdots + \ebc{n-1}{\le D-d}{d}.
\end{equation*}
In light of this identity, in order to achieve a running time of \TEruntime{} we aim to design an algorithm with $d+1$ recursive calls on $n-1$ variables with total degrees~$D, D-1, \dotsc, D-d$.

\begin{algorithm}[t!]
    \caption{\label{algo:TrimmedEval}%
    \TrimmedEval$(P)$}
    \SetKwInOut{Input}{Input}
    \SetKwInOut{Output}{Output}
    \DontPrintSemicolon
    
    \Input{$n$-variate polynomial $P$ of individual degree $d$ and total degree $D$}
    \Output{Evaluations $P(Z_\ell)$ for all $\ell \in \ebc{[n]}{\leq D}{d}$}

    \medskip

    \lIf{$n = 0$}{\Return the constant $P$.}
    
    Write $P(X_1, \dots , X_{n}) = \sum_{i=0}^{d} P_i(X_1, \dots , X_{n-1}) \cdot X_{n}^i$.
    
    Compute the \textsc{LU} decomposition $V(z_{n,0}, \dots , z_{n,d})=L \cdot U$. \label{algo:TrimmedEval:LU}

    Compute the vector of polynomials
    \begin{align*}
        \hspace{3cm}
        \begin{pmatrix} 
            Q_0 \\
            \vdots \\
            Q_d \\
        \end{pmatrix}
        =
        \begin{pmatrix}
            U_{0, 0} & \cdots  & U_{0,d}\\
                        & \ddots & \vdots \\
                &  & U_{d,d} \\
        \end{pmatrix}
        \cdot
        \begin{pmatrix} 
            P_0 \\
            \vdots \\
            P_d \\
        \end{pmatrix}.
    \end{align*}   \label{algo:TrimmedEval:polynomials:Q}
    
    \For{$j = 0,\dots, d$ \label{algo:TrimmedEval:recursive:calls1}}{
        Recursively call \TrimmedEval$(Q_j)$ to evaluate $Q_j(Z_{\ell'})$ for $\ell' \in \ebc{[n-1]}{\le D-j}{d}$.
        \label{algo:TrimmedEval:recursive:calls2}
    }
    
    \For{$\ell' \in \ebc{[n-1]}{\le D}{d}$\label{algo:TrimmedEval:evaluations:1}}{
        Let $k = \min\{d, D - \sum_{j=1}^{n-1} \ell'_j \}$ and compute all evaluations
        \begin{align*}
            \hspace{1.2cm}
            \begin{pmatrix} 
                P\parens*{Z_{(\ell', 0)}} \\
                \vdots \\
                P\parens*{Z_{(\ell', k)}} \\
            \end{pmatrix}
            =
            \begin{pmatrix}
                L_{0,0}   &        & \\
                \vdots    & \ddots & \\
                L_{k, 0} & \cdots & L_{k,k} \\
            \end{pmatrix}
            \cdot
            \begin{pmatrix} 
                Q_0(Z_{\ell'}) \\
                \vdots \\
                Q_{k}(Z_{\ell'}) \\
            \end{pmatrix}.
        \end{align*}  \label{algo:TrimmedEval:evaluations:2}
    }
\end{algorithm}

\subparagraph*{Idea 2: LU Decomposition.}
Our goal is an interleaving of the~$P_i$'s into polynomials~$Q_j$ satisfying two properties:
(1) All evaluations~$Q_j(Z_{\ell'})$ can be recursively computed in a call of size~$\ebc{n-1}{\le D-j}{d}$; in particular, $Q_j$ must have degree~$D-j$.
(2) Simultaneously, we need to recover all evaluations~$P(Z_\ell)$ from the recursively computed evaluations~$Q_j(Z_{\ell'})$ for~$\ell' \in \ebc{[n-1]}{\le D-j}{d}$.
Due to the degree restriction for the polynomials~$Q_j$ in (1) and evaluations~$Q_j(Z_{\ell'})$ in (2), this coincides with a matrix factorization of~$V = L \cdot U$ where $U$ is upper-triangular and $L$ lower-triangular, i.e., a \textsc{LU} decomposition.

In more detail: Why is the \textsc{LU} decomposition useful in keeping the degrees of the polynomials~$Q_j$ ``small''?
Let $p = (P_0, \dots , P_d)^T$ be the vector of polynomials~$P_i$ and consider the product~$U \cdot p = (Q_1, \dotsc , Q_d)^T = q$. 
Here, the upper triangular structure of $U$ guarantees that each~$Q_j = \sum_{i= j}^{d} U_{j,i} \cdot P_i$ has at most degree~$D-j$, since the first non-zero entry in the $j$-th row of $U$ corresponds to polynomial~$P_j$ of degree~$D-j$ while all other $P_i$'s with (possible) non-zero coefficients have a smaller degree.
Additionally, we can also recover~$P(Z_\ell)$ from only the recursively computed~$Q_j(Z_{\ell'})$ as $L$ has lower triangular structure---appropriately exploiting the fact that~$L \cdot q  = L \cdot U \cdot p = V \cdot p$.\footnote{The analogy to van der Hoeven and Schost's algorithm is as follows: Broadly speaking, their algorithm is based on efficient transformations between not only the coefficient and evaluation-based representations of polynomials, but also involving a third representation based on so-called \emph{Newton polynomials.} In this terminology, the matrix $U$ performs a basis change from the monomial basis to the Newton basis, and the matrix $L$ performs a basic change from the Newton basis to the evaluation basis.}
With these ideas in mind, we are now in the position to present the algorithm in detail.

\subsection{Trimmed Multipoint Evaluation}
\label{subsec:trimmed-multipoint-evualation}
Throughout this section, we will prove the following theorem.

\begin{theorem}[Trimmed Multipoint Evaluation]
    \label{thm:trimmed-eval}
    Let $P$ be an $n$-variate polynomial over~$\F$ with individual degree~$d$ and total degree~$D$, and let $Z$ be a grid. 
    The evaluations~$P(Z_\ell)$ can be computed for all $\ell \in \ebc{[n]}{\leq D}{d}$ in time \TEruntime.
\end{theorem}

For the proof of \cref{thm:trimmed-eval}, consider the algorithm \TrimmedEval{}.
Throughout we assume that the algorithm has access to the grid~$Z$ and for simplicity, we omit $Z$ in the recursive calls.

\begin{lemma}[Correctness of \TrimmedEval{}]
    \label{lem:TECorrectness}
    Given an $n$-variate polynomial $P$ over~$\F$ with individual degree $d$ and total degree $D$, \TrimmedEval{} correctly computes the evaluations $P(Z_\ell)$ for all $\ell \in \ebc{[n]}{\leq D}{d}$.
\end{lemma}
\begin{proof}
    Given a polynomial $P$ and a grid $Z$, we show that \TrimmedEval{} correctly computes all evaluations $P(Z_\ell)$ by induction on $n$.
    
    For $n=0$, the polynomial~$P$ does not depend on any variable. 
    Thus, we correctly return the constant~$P$. 
    For $n \ge 1$, we write \makebox{$P(X_1, \dots , X_{n}) = \sum_{i=0}^{d} P_i(X_1, \dots , X_{n-1}) \cdot X_{n}^i$} to obtain $d+1$ many $(n - 1)$-variate polynomials~$P_0, \dots , P_d$ of degrees~$D, \dots , D - d$, respectively.
    Formally, each polynomial~$P_i$ is the coefficient of $P$ viewed as a univariate polynomial in~$X_n$.
    Let $V \coloneqq V(z_{n,0}, \dots , z_{n,d})$ be a Vandermonde matrix.
    It is well-known that~$V$ is invertible if and only if all grid points $z_{n,i}$ are distinct, and in this case the \textsc{LU} decomposition~$V = L \cdot U$ in Line~\ref{algo:TrimmedEval:LU} always exists (see e.g.~\cite{Oruc2000}).
    Further, let $Q_0, \dotsc , Q_d$ as computed in Line~\ref{algo:TrimmedEval:polynomials:Q} of~\TrimmedEval{}.
    Note that each polynomial $Q_j(X_1, \dotsc , X_{n-1})  = \sum_{i=j}^{d} U_{j,i} \cdot P_i(X_1, \dotsc , X_{n-1})$ has degree~$D - j$.
    Indeed, the upper triangular structure of $U$ guarantees that $Q_j$ has at most degree $D-j$, because the first non-zero entry in the $j$-th row of $U$ corresponds to the polynomial $P_j$ of degree~$D-j$ while all other $P_i$'s with (possible) non-zero coefficients have a smaller degree.

    Calling \TrimmedEval$(Q_j)$, we recursively compute the evaluations~$Q_j(Z_{\ell'})$ for all~$\ell' \in \ebc{[n-1]}{\le D-j}{d}$ in Lines~\ref{algo:TrimmedEval:recursive:calls1} and~\ref{algo:TrimmedEval:recursive:calls2}.
    Next, we focus on any iteration $\ell'$ of the loop in Line~\ref{algo:TrimmedEval:evaluations:1}.
    Let~$k = \min\{d, D -  \sum_{j = 1}^{n-1} \ell'_j \}$, then we show that we correctly compute the evaluations~$P(Z_{\ell})$.
    For each $j \in \set{0, \dotsc, k}$, we have
    \begin{align*}
        \sum_{i= 0}^{k} L_{j,i} \cdot Q_i(Z_{\ell'}) 
        &= \sum_{i= 0}^{d} L_{j,i} \cdot Q_i(Z_{\ell'}) \\
        &= \sum_{i= 0}^{d} L_{j,i} \cdot \sum_{m= 0}^{d} U_{i,m} \cdot P_m(Z_{\ell'}) \\
        &= \sum_{i= 0}^{d} z_{n,j}^i \cdot P_i(Z_{\ell'}) \\
        &= P\parens*{Z_{(\ell', j)}},
    \end{align*}
    where the first equality follows from the fact that $L_{j,i} = 0$ whenever $i > k \ge j$.
    Notably, restricting the computation to $k$ is crucial to ensure that we only use the evaluations~$Q_j(Z_{\ell'})$ that we actually computed.
    This proves that \TrimmedEval{} correctly computes the evaluations $P\parens*{Z_{(\ell', 0)}}, \dotsc ,P\parens*{Z_{(\ell', k)}}$ in Line~\ref{algo:TrimmedEval:evaluations:2}.
    Since each $Z_\ell$ can be expressed as some~$Z_{(\ell', j)}$ as above, the loop in Line~\ref{algo:TrimmedEval:evaluations:1} recovers all evaluations~$P(Z_\ell)$.
    As a result, the algorithm \TrimmedEval{} correctly computes all evaluations~$P(Z_\ell)$.
\end{proof}

Lastly, we prove the running time of \TrimmedEval.

\begin{lemma}[Running Time of \TrimmedEval{}]
    \label{lem:TERunningTime}
    The algorithm \TrimmedEval{} runs in time \TEruntime{}.
\end{lemma}
\begin{proof}
    Let $T(n, d, D)$ be the running time of \TrimmedEval{}.
    The time to compute the \textsc{LU} decomposition in Line~$4$ and the matrix-vector products in Lines~$5$ and $9$ can be bounded by $\poly(n,d)$.
    Consequently, the running time of the algorithm \emph{without} the recursive calls in Line~$7$ is $\ebc{n}{\leq D}{d} \cdot M(n,d)$ for some function\footnote{In our context, we are content with bounding $M(n, d) = \poly(n, d)$; however, let us remark that following van der Hoeven and Schost~\cite{HoevenS13} we could achieve a dependence on $d$ that is only polylogarithmic. The main insight is that the L- and U-factors of a Vandermonde matrix are sufficiently structured to support matrix-vector operations in time \smash{$O(d (\log d)^{O(1)})$}.}~$M(n,d) = \poly(n,d)$.
    Thus, the algorithm admits the following recurrence
    \begin{align*}
        T(n,d,D) &\le \ebc{n}{\leq D}{d} \cdot M(n,d) + \sum_{i=0}^d T(n-1, d, D-i).
    \end{align*}
    By induction on $n$, we show that $T(n, d, D) \le \ebc{n}{\leq D}{d} \cdot M(n,d) \cdot n$.
    Indeed, it holds that
    \begin{align*}
        T(n, d, D) 
        &\le  \ebc{n}{\leq D}{d} \cdot M(n,d) + \sum_{i=0}^d T(n-1, d, D-i)\\
        &\le \ebc{n}{\leq D}{d} \cdot M(n,d) + \sum_{i=0}^d \ebc{n}{\leq D-i}{d} \cdot M(n,d)  \cdot(n-1)\\
        &= \ebc{n}{\leq D}{d} \cdot M(n,d) + \sum_{i=0}^d \sum^{D-i}_{j=0} \ebc{n- 1}{j}{d} \cdot M(n,d)\cdot (n-1)\\
        &\le \ebc{n}{\leq D}{d} \cdot M(n,d) + \sum^{D}_{j=0} \ebc{n}{j}{d} \cdot M(n,d) \cdot (n-1)\\
        &= \ebc{n}{\leq D}{d} \cdot M(n,d) + \ebc{n}{\le D}{d} \cdot M(n,d) \cdot (n-1)\\
        &= \ebc{n}{\le D}{d} \cdot M(n,d) \cdot n
    \end{align*}
    Therefore, \TrimmedEval{} runs in time \TEruntime.
\end{proof}

Combining \cref{lem:TECorrectness} and \cref{lem:TERunningTime} concludes the proof of \cref{thm:trimmed-eval}.

\begin{algorithm}[t!]
    \caption{\label{algo:TrimmedInt}%
    \TrimmedInt$((\alpha_\ell)_\ell)$}
    \SetKwInOut{Input}{Input}
    \SetKwInOut{Output}{Output}
    \DontPrintSemicolon
    
    \Input{Evaluations $\alpha_\ell$ for all $\ell \in \ebc{[n]}{\leq D}{d}$}
    \Output{Unique $n$-variate polynomial $P$ of individual degree $d$ and total degree $D$ such that $P(Z_\ell) = \alpha_\ell$ for all $\ell \in \ebc{[n]}{\leq D}{d}$}

    \medskip
    
    \lIf{$n = 0$}{\Return the constant $P$. \label{algo:TrimmedInt:line:q}}
    Compute the \textsc{UL} decomposition $V(z_{n,0}, \dots, z_{n,d})^{-1} = U \cdot L$. \label{algo:TrimmedInt:LU} \;
    \For{$\ell' \in \ebc{[n-1]}{\le D}{d}$}{
        Let $k = \min\{d, D -  \sum_{j = 1}^{n-1} \ell'_j\}$ and compute all evaluations
        \begin{align*}
            \hspace{1.25cm}
            \begin{pmatrix} 
                \beta_{(\ell', 0)} \\
                \vdots \\
                \beta_{(\ell', k)} \\
            \end{pmatrix}
            =
            \begin{pmatrix}
                L_{0, 0} &        &         \\
                \vdots   & \ddots &         \\
                L_{k, 0} & \cdots & U_{k,k} \\
            \end{pmatrix}
            \cdot
            \begin{pmatrix} 
                \alpha_{(\ell', 0)} \\
                \vdots \\
                \alpha_{(\ell', k)} \\
            \end{pmatrix}. \;
        \end{align*}
    }

    \For{$j = 0,\dots, d$}{
        Recursively call \TrimmedInt$\parens*{(\beta_{(\ell', j)})_{\ell'}}$ for ${\ell'} \in \ebc{[n-1]}{\le D-j}{d}$ to interpolate the $(n-1)$-variate polynomial $Q_j$ (of total degree $D - j$). \label{algo:TrimmedInt:recurse}\;
    }

    Compute the vector of polynomials
        \begin{align*}
            \hspace{2.3cm}
            \begin{pmatrix} 
                P_0 \\
                \vdots \\
                P_d \\
            \end{pmatrix}
            =
            \begin{pmatrix}
                U_{0, 0} & \cdots & U_{0, d} \\
                         & \ddots & \vdots   \\
                         &        & U_{d, d} \\
            \end{pmatrix}
            \cdot
            \begin{pmatrix} 
                Q_0 \\
                \vdots \\
                Q_d \\
            \end{pmatrix}. \; 
        \end{align*}

    Compute $P(X_1, \dots , X_{n}) = \sum_{i=0}^{d} P_i(X_1, \dots , X_{n-1}) \cdot X_{n}^i$. \;
\end{algorithm}

\subsection{Trimmed Interpolation}

Throughout this section, we show that a polynomial can be (uniquely) interpolated from its evaluations on the trimmed grid points---using essentially the same algorithmic approach as in \cref{subsec:trimmed-multipoint-evualation}.

\begin{theorem}[Trimmed Interpolation]
    \label{thm:trimmed-int}
    Let $\alpha_\ell \in \mathbb F$ for~$\ell \in \ebc{[n]}{\leq D}{d}$, and let~$Z$ be a grid.  
    The unique $n$-variate polynomial $P$ with individual degree~$d$ and total degree~$D$ that satisfies $P(Z_\ell) = \alpha_\ell$ for all~$\ell \in \ebc{[n]}{\leq D}{d}$ can be computed in time \TEruntime.
\end{theorem}

Consider the algorithm \TrimmedInt{} for the proof of \cref{thm:trimmed-int}.
As before, we assume that the algorithm has access to the grid~$Z$ and for simplicity, we omit $Z$ in the recursive calls.

\begin{lemma}[Correctness of \TrimmedInt{}]
    \label{lem:TICorrectness}
   Given evaluations $\alpha_\ell$ for~$\ell \in \ebc{[n]}{\leq D}{d}$, \TrimmedInt{} correctly interpolates the unique $n$-variate polynomial~$P$ with individual degree $d$ and total degree $D$ such that $P(Z_\ell) = \alpha_\ell$.
\end{lemma}
\begin{proof}
    Given the evaluations $\alpha_\ell$ for all~$\ell \in \ebc{[n]}{\leq D}{d}$, we show that \TrimmedInt{} correctly interpolates the unique $n$-variate polynomial $P$ by induction on~$n$.
    If $n=0$, simply return the constant~$P$ as it does not depend on any variable.
    For the remainder of the proof we assume $n\ge 1$. The induction hypothesis implies that we have $Q_j(Z_\ell') = \beta_{(\ell', j)}$ as defined in Line~\ref{algo:TrimmedInt:recurse}, for all $\ell' \in \ebc{[n-1]}{\leq D-j}{d}$.
    
    Let $V \coloneqq V(z_{n,0}, \dots , z_{n,d})$.
    It is well-known that~$V$ is invertible if and only if the grid points~$z_{n,0}, \dots, z_{n,d}$ are distinct.
    Moreover, in this case the \textsc{UL} decomposition $V^{-1} = U \cdot L$ from Line~\ref{algo:TrimmedInt:LU} always exists (see e.g.~\cite{Oruc2000}). As the matrices $U$ and $L$ are invertible, too, in the following calculation we will use the equivalent formulation $V \cdot U = L^{-1}$.
    
    We now show that $P(Z_\ell) = \alpha_\ell$ for all~\smash{$\ell \in \ebc{[n]}{\leq D}{d}$}. Note that each~$Z_\ell$ can be expressed as some~$Z_{(\ell',j)}$ with \smash{$\ell' \in \ebc{[n-1]}{\le D-j}{d}$} and $\sum_{i=1}^{n-1} \ell'_i + j \leq D$.
    Fix such a pair~$\ell'$ and $j$. As in the algorithm, let $k = \min\{d, D - \sum_{i=0}^{n-1} \ell_i'\}$. Clearly we have $j \leq k$. It holds that:
    \begin{align*}
        P\parens*{Z_{(\ell', j)}}
        &= \sum_{i=0}^{d} z_{n,j}^i \cdot P_i(Z_{\ell'}) \\
        &= \sum_{i=0}^{d} V_{j, i} \cdot \sum_{m=0}^{d} U_{i, m} \cdot Q_m(Z_{\ell'})  \\
        &= \sum_{i=0}^{d} V_{j, i} \cdot \sum_{m=i}^{d} U_{i, m} \cdot Q_m(Z_{\ell'})  \\
        &= \sum_{m=0}^{d} Q_m(Z_{\ell'}) \cdot \sum_{i=0}^m V_{j, i} \cdot U_{i, m} \\
        &= \sum_{m=0}^{d} Q_m(Z_{\ell'}) \cdot (L^{-1})_{j, m} \\
        &= \sum_{m=0}^{j} Q_m(Z_{\ell'}) \cdot (L^{-1})_{j, m} \\
        &= \sum_{m=0}^{j} \beta_{(\ell', m)} \cdot (L^{-1})_{j, m} \\
        &= \sum_{m=0}^{k} \beta_{(\ell', m)} \cdot (L^{-1})_{j, m} \\
        &= \alpha_{(\ell', j)}.
    \end{align*}
    Consequently, \TrimmedInt{} correctly interpolates the unique $n$-variate polynomial~$P$ with individual degree $d$ and total degree $D$ such that~$P(Z_\ell) = \alpha_\ell$.
\end{proof}
    
As the recursion scheme of \TrimmedInt{} and \TrimmedEval{} are exactly the same, we refer to \cref{lem:TERunningTime} for a proof of the running time of \TrimmedInt. 
This concludes the proof of \cref{thm:trimmed-int}.

\bibliographystyle{plainurl}
\bibliography{refs}

@inproceedings{Umans08,
  author       = {Christopher Umans},
  editor       = {Cynthia Dwork},
  title        = {Fast polynomial factorization and modular composition in small characteristic},
  booktitle    = {40th Annual {ACM} Symposium on Theory of Computing ({STOC} 2008)},
  pages        = {481--490},
  publisher    = {{ACM}},
  year         = {2008},
  url          = {https://doi.org/10.1145/1374376.1374445},
  doi          = {10.1145/1374376.1374445},
  bibsource    = {dblp computer science bibliography, https://dblp.org}
}

@inproceedings{BhargavaG0M22,
  author       = {Vishwas Bhargava and
                  Sumanta Ghosh and
                  Mrinal Kumar and
                  Chandra Kanta Mohapatra},
  editor       = {Stefano Leonardi and
                  Anupam Gupta},
  title        = {Fast, algebraic multivariate multipoint evaluation in small characteristic
                  and applications},
  booktitle    = {54th Annual {ACM} Symposium on Theory of Computing ({STOC} 2022)},
  pages        = {403--415},
  publisher    = {{ACM}},
  year         = {2022},
  url          = {https://doi.org/10.1145/3519935.3519968},
  doi          = {10.1145/3519935.3519968},
  bibsource    = {dblp computer science bibliography, https://dblp.org}
}

@inproceedings{BjorklundH13,
  author       = {Andreas Bj{\"{o}}rklund and
                  Thore Husfeldt},
  title        = {The Parity of Directed Hamiltonian Cycles},
  booktitle    = {54th Annual {IEEE} Symposium on Foundations of Computer Science, {FOCS}
                  2013, 26-29 October, 2013, Berkeley, CA, {USA}},
  pages        = {727--735},
  publisher    = {{IEEE} Computer Society},
  year         = {2013},
  url          = {https://doi.org/10.1109/FOCS.2013.83},
  doi          = {10.1109/FOCS.2013.83},
  bibsource    = {dblp computer science bibliography, https://dblp.org}
}

@inproceedings{BhargavaGG0U22,
  author       = {Vishwas Bhargava and
                  Sumanta Ghosh and
                  Zeyu Guo and
                  Mrinal Kumar and
                  Chris Umans},
  title        = {Fast Multivariate Multipoint Evaluation Over All Finite Fields},
  booktitle    = {63rd {IEEE} Annual Symposium on Foundations of Computer Science ({FOCS} 2022)},
  pages        = {221--232},
  publisher    = {{IEEE}},
  year         = {2022},
  url          = {https://doi.org/10.1109/FOCS54457.2022.00028},
  doi          = {10.1109/FOCS54457.2022.00028},
  bibsource    = {dblp computer science bibliography, https://dblp.org}
}

@inproceedings{LokshtanovPTWY17,
  author       = {Daniel Lokshtanov and
                  Ramamohan Paturi and
                  Suguru Tamaki and
                  R. Ryan Williams and
                  Huacheng Yu},
  editor       = {Philip N. Klein},
  title        = {Beating Brute Force for Systems of Polynomial Equations over Finite
                  Fields},
  booktitle    = {28th Annual {ACM-SIAM} Symposium on Discrete Algorithms ({SODA} 2017)},
  pages        = {2190--2202},
  publisher    = {{SIAM}},
  year         = {2017},
  url          = {https://doi.org/10.1137/1.9781611974782.143},
  doi          = {10.1137/1.9781611974782.143},
  bibsource    = {dblp computer science bibliography, https://dblp.org}
}

@inproceedings{Dinur21,
  author       = {Itai Dinur},
  editor       = {D{\'{a}}niel Marx},
  title        = {Improved Algorithms for Solving Polynomial Systems over {GF(2)} by
                  Multiple Parity-Counting},
  booktitle    = {32nd Annual {ACM-SIAM} Symposium on Discrete Algorithms ({SODA} 2021)},
  pages        = {2550--2564},
  publisher    = {{SIAM}},
  year         = {2021},
  url          = {https://doi.org/10.1137/1.9781611976465.151},
  doi          = {10.1137/1.9781611976465.151},
  bibsource    = {dblp computer science bibliography, https://dblp.org}
}

@inproceedings{DellHKW25,
  author       = {Holger Dell and
                  Anselm Haak and
                  Melvin Kallmayer and
                  Leo Wennmann},
  editor       = {Yossi Azar and
                  Debmalya Panigrahi},
  title        = {Solving Polynomial Equations Over Finite Fields},
  booktitle    = {36th Annual {ACM-SIAM} Symposium on Discrete Algorithms ({SODA} 2025)},
  pages        = {2779--2803},
  publisher    = {{SIAM}},
  year         = {2025},
  url          = {https://doi.org/10.1137/1.9781611978322.90},
  doi          = {10.1137/1.9781611978322.90},
  bibsource    = {dblp computer science bibliography, https://dblp.org}
}

@inproceedings{BjorklundK019,
  author       = {Andreas Bj{\"{o}}rklund and
                  Petteri Kaski and
                  Ryan Williams},
  editor       = {Christel Baier and
                  Ioannis Chatzigiannakis and
                  Paola Flocchini and
                  Stefano Leonardi},
  title        = {Solving Systems of Polynomial Equations over {GF(2)} by a Parity-Counting
                  Self-Reduction},
  booktitle    = {46th International Colloquium on Automata, Languages, and Programming ({ICALP} 2019)},
  series       = {LIPIcs},
  volume       = {132},
  pages        = {26:1--26:13},
  publisher    = {Schloss Dagstuhl - Leibniz-Zentrum f{\"{u}}r Informatik},
  year         = {2019},
  url          = {https://doi.org/10.4230/LIPIcs.ICALP.2019.26},
  doi          = {10.4230/LIPICS.ICALP.2019.26},
  bibsource    = {dblp computer science bibliography, https://dblp.org}
}

@inproceedings{NuskenZ04,
  author       = {Michael N{\"{u}}sken and
                  Martin Ziegler},
  editor       = {Susanne Albers and
                  Tomasz Radzik},
  title        = {Fast Multipoint Evaluation of Bivariate Polynomials},
  booktitle    = {12th Annual European Symposium on Algorithms ({ESA} 2004)},
  series       = {Lecture Notes in Computer Science},
  volume       = {3221},
  pages        = {544--555},
  publisher    = {Springer},
  year         = {2004},
  url          = {https://doi.org/10.1007/978-3-540-30140-0\_49},
  doi          = {10.1007/978-3-540-30140-0\_49},
  bibsource    = {dblp computer science bibliography, https://dblp.org}
}

@article{BjorklundKW19,
  author       = {Andreas Bj{\"{o}}rklund and
                  Petteri Kaski and
                  Ryan Williams},
  title        = {Generalized Kakeya sets for polynomial evaluation and faster computation
                  of fermionants},
  journal      = {Algorithmica},
  volume       = {81},
  number       = {10},
  pages        = {4010--4028},
  year         = {2019},
  url          = {https://doi.org/10.1007/s00453-018-0513-7},
  doi          = {10.1007/S00453-018-0513-7},
  bibsource    = {dblp computer science bibliography, https://dblp.org}
}

@inproceedings{Patarin96,
  author       = {Jacques Patarin},
  editor       = {Ueli M. Maurer},
  title        = {Hidden Fields Equations {(HFE)} and Isomorphisms of Polynomials {(IP):}
                  Two New Families of Asymmetric Algorithms},
  booktitle    = {International Conference on the Theory and Application of Cryptographic Techniques ({EUROCRYPT} 1996)},
  series       = {Lecture Notes in Computer Science},
  volume       = {1070},
  pages        = {33--48},
  publisher    = {Springer},
  year         = {1996},
  url          = {https://doi.org/10.1007/3-540-68339-9\_4},
  doi          = {10.1007/3-540-68339-9\_4},
  bibsource    = {dblp computer science bibliography, https://dblp.org}
}

@inproceedings{KipnisPG99,
  author       = {Aviad Kipnis and
                  Jacques Patarin and
                  Louis Goubin},
  editor       = {Jacques Stern},
  title        = {Unbalanced Oil and Vinegar Signature Schemes},
  booktitle    = {International Conference on the Theory and Application of Cryptographic Techniques ({EUROCRYPT} 1999)},
  series       = {Lecture Notes in Computer Science},
  volume       = {1592},
  pages        = {206--222},
  publisher    = {Springer},
  year         = {1999},
  url          = {https://doi.org/10.1007/3-540-48910-X\_15},
  doi          = {10.1007/3-540-48910-X\_15},
  bibsource    = {dblp computer science bibliography, https://dblp.org}
}

@article{KedlayaU11,
  author       = {Kiran S. Kedlaya and
                  Christopher Umans},
  title        = {Fast Polynomial Factorization and Modular Composition},
  journal      = {{SIAM} J. Comput.},
  volume       = {40},
  number       = {6},
  pages        = {1767--1802},
  year         = {2011},
  url          = {https://doi.org/10.1137/08073408X},
  doi          = {10.1137/08073408X},
  bibsource    = {dblp computer science bibliography, https://dblp.org}
}

@article{HoevenL21a,
  author       = {Joris van der Hoeven and
                  Gr{\'{e}}goire Lecerf},
  title        = {Fast amortized multi-point evaluation},
  journal      = {J. Complex.},
  volume       = {67},
  pages        = {101574},
  year         = {2021},
  url          = {https://doi.org/10.1016/j.jco.2021.101574},
  doi          = {10.1016/J.JCO.2021.101574},
  bibsource    = {dblp computer science bibliography, https://dblp.org}
}

@article{Yates37,
  title        = {The design and analysis of factorial experiments},
  author       = {Frank Yates},
  year         = {1937},
  publisher    = {Imperial Bureau of Soil Science}
}

@article{BorodinM74,
  author       = {Allan Borodin and
                  R. Moenck},
  title        = {Fast Modular Transforms},
  journal      = {J. Comput. Syst. Sci.},
  volume       = {8},
  number       = {3},
  pages        = {366--386},
  year         = {1974},
  url          = {https://doi.org/10.1016/S0022-0000(74)80029-2},
  doi          = {10.1016/S0022-0000(74)80029-2},
  bibsource    = {dblp computer science bibliography, https://dblp.org}
}

@article{Werner80,
  author       = {Helmut Werner},
  title        = {Remarks on Newton type multivariate interpolation for subsets of grids},
  journal      = {Computing},
  volume       = {25},
  number       = {2},
  pages        = {181--191},
  year         = {1980},
  url          = {https://doi.org/10.1007/BF02259644},
  doi          = {10.1007/BF02259644},
  bibsource    = {dblp computer science bibliography, https://dblp.org}
}

@article{Boor92,
  author       = {Carl De Boor and
                  Amos Ron},
  title        = {Computational Aspects of Polynomial Interpolation in Several Variables},
  journal      = {Mathematics of Computation},
  volume       = {58},
  number       = {198},
  pages        = {705},
  year         = {1992},
  issn         = {0025--5718},
  url          = {http://dx.doi.org/10.2307/2153210},
  doi          = {10.2307/2153210},
}

@article{Bollinger93,
  author       = {Richard C. Bollinger},
  title        = {Extended Pascal Triangles},
  journal      = {Mathematics Magazine},
  volume       = {66},
  number       = {2},
  pages        = {87--94},
  year         = {1993},
  issn         = {1930-0980},
  url          = {http://dx.doi.org/10.1080/0025570X.1993.11996088},
  doi          = {10.1080/0025570x.1993.11996088},
}

@article{Pan94,
  author       = {Victor Y. Pan},
  title        = {Simple Multivariate Polynomial Multiplication},
  journal      = {J. Symb. Comput.},
  volume       = {18},
  number       = {3},
  pages        = {183--186},
  year         = {1994},
  url          = {https://doi.org/10.1006/jsco.1994.1042},
  doi          = {10.1006/JSCO.1994.1042},
  bibsource    = {dblp computer science bibliography, https://dblp.org}
}

@article{Oruc2000,
  author       = {Halil Oru\c{c} and
                  George M. Phillips},
  title        = {Explicit factorization of the Vandermonde matrix},
  journal      = {Linear Algebra and its Applications},
  volume       = {315},
  number       = {1--3},
  pages        = {113--123},
  year         = {2000},
  url          = {http://dx.doi.org/10.1016/S0024-3795(00)00124-5},
  doi          = {10.1016/s0024-3795(00)00124-5},
}

@article{Sauer04,
  author       = {Tomas Sauer},
  title        = {Lagrange interpolation on subgrids of tensor product grids},
  journal      = {Math. Comput.},
  volume       = {73},
  number       = {245},
  pages        = {181--190},
  year         = {2004},
  url          = {https://doi.org/10.1090/S0025-5718-03-01557-6},
  doi          = {10.1090/S0025-5718-03-01557-6},
  bibsource    = {dblp computer science bibliography, https://dblp.org}
}

@article{BjorklundHKK10,
  author       = {Andreas Bj{\"{o}}rklund and
                  Thore Husfeldt and
                  Petteri Kaski and
                  Mikko Koivisto},
  title        = {Trimmed Moebius Inversion and Graphs of Bounded Degree},
  journal      = {Theory Comput. Syst.},
  volume       = {47},
  number       = {3},
  pages        = {637--654},
  year         = {2010},
  url          = {https://doi.org/10.1007/s00224-009-9185-7},
  doi          = {10.1007/S00224-009-9185-7},
  bibsource    = {dblp computer science bibliography, https://dblp.org}
}

@article{HoevenS13,
  author       = {Joris van der Hoeven and
                  {\'{E}}ric Schost},
  title        = {Multi-point evaluation in higher dimensions},
  journal      = {Appl. Algebra Eng. Commun. Comput.},
  volume       = {24},
  number       = {1},
  pages        = {37--52},
  year         = {2013},
  url          = {https://doi.org/10.1007/s00200-012-0179-3},
  doi          = {10.1007/S00200-012-0179-3},
  bibsource    = {dblp computer science bibliography, https://dblp.org}
}

@article{ChkifaCS14,
  author       = {Abdellah Chkifa and
                  Albert Cohen and
                  Christoph Schwab},
  title        = {High-Dimensional Adaptive Sparse Polynomial Interpolation and Applications
                  to Parametric PDEs},
  journal      = {Found. Comput. Math.},
  volume       = {14},
  number       = {4},
  pages        = {601--633},
  year         = {2014},
  url          = {https://doi.org/10.1007/s10208-013-9154-z},
  doi          = {10.1007/S10208-013-9154-Z},
  bibsource    = {dblp computer science bibliography, https://dblp.org}
}

@article{DynF14,
  author       = {Nira Dyn and
                  Michael S. Floater},
  title        = {Multivariate polynomial interpolation on lower sets},
  journal      = {J. Approx. Theory},
  volume       = {177},
  pages        = {34--42},
  year         = {2014},
  url          = {https://doi.org/10.1016/j.jat.2013.09.008},
  doi          = {10.1016/J.JAT.2013.09.008},
  bibsource    = {dblp computer science bibliography, https://dblp.org}
}

@book{BiniP94,
  author       = {Dario Bini and
                  Victor Y. Pan},
  title        = {Polynomial and matrix computations, 1st Edition},
  series       = {Progress in theoretical computer science},
  volume       = {12},
  publisher    = {Birkh{\"{a}}user},
  year         = {1994},
  url          = {https://www.worldcat.org/oclc/312012822},
  isbn         = {3764337869},
  bibsource    = {dblp computer science bibliography, https://dblp.org}
}

@book{vonzurGathenG13,
  author    = {Joachim von zur Gathen and Jürgen Gerhard},
  title     = {Modern Computer Algebra},
  publisher = {Cambridge University Press},
  edition   = {3rd},
  year      = {2013},
  doi       = {10.1017/CBO9781139856065},
}

\end{document}